\documentclass{article}
\usepackage{amsmath}
\usepackage{amsthm}

\usepackage{hyperref}
\usepackage{cleveref}

\usepackage{enumerate}
\usepackage{fullpage}

\renewcommand{\epsilon}{\varepsilon}
\newtheorem*{claim*}{Claim}

\title{There is no APTAS for 2-Dimensional Vector Bin Packing: Revisited}

\author{Arka Ray\\ Indian Institute of Science, Bengaluru\\ arkaray@iisc.ac.in}
\date{}

\newcommand{\opt}{\texttt{OPT}}

\newtheorem{theorem}{Theorem}[section]

\newtheorem{lemma}[theorem]{Lemma}

\newtheorem{observation}[theorem]{Observation}

\begin{document}
	
	\maketitle
 \begin{abstract}
    We study the Vector Bin Packing and the Vector Bin Covering problems, multidimensional generalizations of the Bin Packing and the Bin Covering problems, respectively.
In the Vector Bin Packing, we are given a set of $d$-dimensional vectors from $[0,1]^d$ and the aim is to partition the set into the minimum number of bins such that for each bin $B$, each component of the sum of the vectors in $B$ is at most 1.
Woeginger \cite{DBLP:journals/ipl/Woeginger97} claimed that the problem has no APTAS for dimensions greater than or equal to 2.
We note that there was a slight oversight in the original proof.
In this work, we give a revised proof using some additional ideas from \cite{DBLP:journals/mor/BansalCKS06,DBLP:journals/jda/ChlebikC09}.
In fact, we show that it is NP-hard to get an asymptotic approximation ratio better than $\frac{600}{599}$.

An instance of Vector Bin Packing is called $\delta$-skewed if every item has at most one dimension greater than $\delta$.
As a natural extension of our general $d$-Dimensional Vector Bin Packing result we show that for $\epsilon\in (0,\frac{1}{2500})$ it is NP-hard to obtain a $(1+\epsilon)$-approximation for $\delta$-Skewed Vector Bin Packing if $\delta>20\sqrt \epsilon$.

In the Vector Bin Covering problem given a set of $d$-dimensional vectors from $[0,1]^d$, the aim is to obtain a family of disjoint subsets (called bins) with the maximum cardinality such that for each bin $B$, each component of the sum of the vectors in $B$ is at least 1.
Using ideas similar to our Vector Bin Packing result, we show that for Vector Bin Covering there is no APTAS for dimensions greater than or equal to 2.
In fact, we show that it is NP-hard to get an asymptotic approximation ratio better than $\frac{998}{997}$.
	
\end{abstract}

	\section{Introduction}
In the $d$-Dimensional Vector Bin Packing problem, we are given a set of $d$-dimensional vectors (say $S$), each of whose components belongs to $[0,1]$, i.e., $S\subseteq [0,1]^d$.
The aim is to partition $S$ into the minimum number of bins, such that for each bin $B$, each component of the sum of the vectors in $B$ is at most 1.
Whenever the above condition holds for a set of vectors $B$, we say that the vectors in $B$ fit in a bin.
The Vector Bin Packing problem is a natural generalization of the Bin Packing problem, which can be obtained by setting $d=1$.

We also study a related problem called the Vector Bin Covering problem, a generalization of the Bin Covering problem.
In the $d$-Dimensional Vector Bin Covering problem, we are again given a set of $d$-dimensional vectors (say $S$), each of whose components belongs to $[0,1]$.
The aim is to obtain a family of disjoint subsets (these subsets are called bins) with the maximum cardinality such that for each bin $B$, each component of the sum of the vectors in $B$ is at least 1.
	
There is a well-known reduction from the Partition problem\footnote{In the Partition problem given a list of positive integers $x_1,\dots,x_n$, the aim is to determine whether there is a subset $S\subseteq \{1,\dots,n\}$ such that $\sum_{i\in S}x_i=\sum_{i\not \in S}x_i$.} to the Bin Packing problem, showing that it is NP-hard to obtain a $3/2$ absolute approximation for Bin Packing.
In fact, the same reduction shows that it is NP-hard to get an absolute approximation ratio of $2-\epsilon$ for Bin Covering.
Hence, we will look at the asymptotic approximation ratio for both of these problems. 

Finally, we also study the Vector Bin Packing problem restricted to skewed instances, i.e., instances where all items are $\delta$-skewed for some constant $\delta>0$.
An item is called $\delta$-\emph{large} if at least two dimensions are larger than $\delta$, for some constant $\delta>0$; otherwise, it is called $\delta$-\emph{skewed}.
In fact, the case where all items are skewed forms an important subcase for many packing problems~\cite{DBLP:conf/soda/BansalE016, DBLP:conf/approx/Galvez0AJ0R20, DBLP:conf/approx/0001S21}.

\subsection{Related Works}
For the Vector Bin Packing problem, when $d$ has been supplied as part of the input, Fernandez de la Vega and Lueker gave a $(d+\epsilon)$-approximate algorithm in \cite{DBLP:journals/combinatorica/VegaL81}.
This algorithm is almost optimal as there is a well-known reduction from the Vertex Coloring problem, which shows a $d^{1-\epsilon}$ hardness (see~\cite{DBLP:conf/soda/BansalE016}).
	
If $d$ is kept fixed, i.e., it is not supplied as part of the input, then the above lower bound does not hold, and in fact, much better approximation factors are known for this case.
The barrier of $d$ was broken by Chekuri and Khanna~\cite{DBLP:journals/siamcomp/ChekuriK04} by obtaining a $\ln d + 2 + \gamma$ approximation, where $\gamma\approx 0.57721$ is the Euler-Mascheroni constant.
This was further improved to $\ln d + 1$ by Bansal, Caprara, and Sviridenko~\cite{DBLP:journals/siamcomp/BansalCS09} and then to $\ln (d+1) + 0.807$ by Bansal, Eli\'a\v{s}, and Khan \cite{DBLP:conf/soda/BansalE016}.
Recently, Sandeep \cite{DBLP:conf/focs/Sandeep21} showed that the best approximation ratio any algorithm which solves Vector Bin Packing can have, for high enough dimensions, is $\Omega (\ln d)$.
In the $d=2$ case, Bansal, Caprara, and Sviridenko \cite{DBLP:journals/siamcomp/BansalCS09} gave a 1.693 approximation algorithm which was later improved to 1.406 by Bansal, Eli\'a\v{s}, and Khan \cite{DBLP:conf/soda/BansalE016}.
Recently, this was further improved to $\frac{4}{3}+\varepsilon$ by Kulik, Mnich, and Shachnai~\cite{DBLP:journals/corr/abs-2205-12828}.

Bansal, Eli\'a\v{s}, and Khan \cite{DBLP:conf/soda/BansalE016} also note that skewed instances %(i.e., $\delta$-skewed instances for some $\delta>0$)
constitute the hard instances for rounding-based algorithms for 2-Dimensional Vector Bin Packing.
G\'alvez, Grandoni, Ameli, Jansen, Khan, and Rau~\cite{DBLP:conf/approx/Galvez0AJ0R20} studied the Strip Packing problem in this context and gave a $(\frac{3}{2}+\epsilon)$-approximation. They also showed an (almost) matching $(\frac{3}{2}-\epsilon)$ lower bound.
Recently, Khan and Sharma \cite{DBLP:conf/approx/0001S21} gave an APTAS for 2-Dimensional Geometric Bin Packing with skewed items.
They also note that it is possible to solve the Maximum Independent Set of Rectangles problem and the 2-Dimensional Geometric Knapsack problem exactly in polynomial time if all items are $\delta$-large.
	
For the Vector Bin Covering problem, when $d$ is supplied as part of the input, the algorithm in the work by Alon, Azar, Csirik, Epstein, Sevastianov, Vestjens, and Woeginger \cite{DBLP:journals/algorithmica/AlonACESVW98} gives an approximation ratio of $O(\ln d)$.
Finally, Sandeep \cite{DBLP:conf/focs/Sandeep21} also gave a lower bound of $\Omega(\frac{\log d}{\log \log d})$ when the run-time is allowed to have a superpolynomial dependence on $d$. 
	
For further information on approximation and online algorithms for multidimensional variants of the Bin Packing and Bin Covering problems, we refer the reader to the survey \cite{DBLP:journals/csr/ChristensenKPT17} by Christensen, Khan, Pokutta, and Tetali.

\subsection{Our Results}
It was believed that \cite{DBLP:journals/ipl/Woeginger97} showed that there is no APTAS for the $d$-Dimensional Vector Bin Packing problem with $d\geq 2$.
However, as we show in \Cref{sec:original}, there was a minor oversight in the original proof.
Unfortunately, this oversight is also present in the $\frac{391}{390}$ lower bound for Vector Bin Packing by Chleb\'ik and Chlebikov\'a \cite{DBLP:journals/eccc/ECCC-TR06-019}.
Hence, we present a revised proof in \Cref{sec:revised}.
Our proof uses essentially the same construction as the original proof.
However, the analysis is slightly different and the main ideas for the analysis are borrowed from \cite{DBLP:journals/mor/BansalCKS06,DBLP:journals/jda/ChlebikC09}.
We note that Sandeep's lower bound of $\Omega(\ln d)$ does not hold for low dimensions, and hence, it does not even rule out the possibility of APTAS in the 2-dimensional case.
	
Our second result concerns Vector Bin Packing with skewed items.
Extending the proof of our non-existence of APTAS for Vector Bin Packing we show that for $\epsilon\in (0,\frac{1}{2500})$ we need $\delta\leq 20\sqrt \epsilon$ to obtain a $(1+\epsilon)$-approximation for $\delta$-Skewed $d$-Dimensional Vector Bin Packing.
Finally, we also show that there is no APTAS for Vector Bin Covering with dimension $d\geq 2$.

\subsection{Preliminaries}
As is the case in \cite{DBLP:journals/ipl/Woeginger97} and \cite{DBLP:journals/mor/BansalCKS06}, we start with Maximum 3-Dimensional Matching (denoted by MAX-3-DM) and reduce it to 4-Partition and then reduce it to Vector Bin Packing.
A 3-Dimensional Matching instance has three sets $X=\{x_1,x_2,\dots,x_q\}$,$Y=\{y_1,y_2,\dots,y_q\}$, and $Z=\{z_1,z_2,\dots,z_q\}$ and a set of tuples $T\subseteq X\times Y\times Z$.
In MAX-3-DM given such an instance the aim is to find a subset $T'\subseteq T$ with the maximum cardinality such that no element from $X,Y,$ or $Z$ occurs in more than one tuple.
For our reduction we consider a restricted variant of this problem where there are exactly 2 tuples containing each element of the sets $X$,$Y$, and $Z$.
This variant is called the 2-Exact Maximum 3-Dimensional Matching problem (denoted by MAX-3-DM-E2).
More precisely, we consider the gap variant of MAX-3-DM-E2 (denoted GAP($\alpha, \beta$)-3-DM-E2),
where given a MAX-3-DM-E2 instance $I_M$ the aim is to distinguish between the case with $\opt(I_M)\geq \lceil \beta q \rceil$ and $\opt(I_M)\leq \lfloor \alpha q \rfloor$, where $\opt(I_M)$ is the optimal solution to the corresponding MAX-3-DM-E2 problem.
In the $m$-Partition problem given a list of integers $x_1,\dots,x_n$ with $n$ being a multiple of $m$, the aim is to determine whether there exists $n/m$ disjoint subsets $S_i\subseteq \{1,\dots, n\}$ of cardinality $m$ such that $\sum_{k \in S_i}x_k=\sum_{k\in S_j}x_k$ for each $i,j\in \{1,\dots,n/m\}$.

Berman and Karpinski \cite{DBLP:journals/eccc/ECCC-TR03-008} showed that it is NP-hard to approximate MAX-3-DM-E2 with ratio better than $\frac{98}{97}$.
This bound was later improved by Chleb\'ik and Chleb\'ikov\'a \cite{DBLP:journals/tcs/ChlebikC06} to $\frac{95}{94}$.
Finally, Chleb\'ik and Chleb\'ikov\'a \cite{DBLP:journals/jda/ChlebikC09} also note the following corollary of their $\frac{95}{94}$ bound.
\begin{theorem}[\cite{DBLP:journals/jda/ChlebikC09}]
\label{thm:matching}
	GAP($\alpha_0, \beta_0$)-3-DM-E2 is NP-hard, where $\alpha_0 = 0.9690082645$ and $\beta_0 = 0.979338843$.
\end{theorem}
For the skewed item case we note that the size of the items in the reduction from MAX-3-DM to 2-Dimensional Vector Bin Packing can be made smaller by going through $m$-Partition instead of 4-Partition.
Finally, for our Vector Bin Covering result we make a minor modification to the Vector Bin Packing reduction.

	\section{The Main Result}
	\label{sec:revised}
	In this section, we prove our main result, i.e., there is no APTAS for Vector Bin Packing.
	We do so by modifying the construction in the original proof given in \cite{DBLP:journals/ipl/Woeginger97} by adding a set of dummy vectors.
	%\footnote{We have tried to stay as close to the original proof as possible in terms of notations and arbitrary choices. Yet we have made two notable changes, (i) using $r=64q$, and (ii) using $t_{(i,j,k)}$ to denote a tuple.}
	The final analysis is based on the analysis in \cite{DBLP:journals/mor/BansalCKS06} for the Geometric Bin Packing lower bound.

	We start by defining a few integers based on the given MAX-3-DM instance $I_M$. Let $r=64q$, where $q=|X|=|Y|=|Z|$ and $b=r^4+15$. Define integers $x'_i,y'_i,z'_i$ corresponding to ${x_i\in X},{y_i\in Y}, {z_i\in Z}$ to be
	\begin{align*}
		x'_i=ir+1,\\
		y'_i=ir^2+2,\\
		z'_i=ir^3+4,
	\end{align*}
	and for each $t_{(i,j,k)}=(x_i,y_j,z_k)\in T$ define $t'_{(i,j,k)}$ as
	\[
		t'_{(i,j,k)}=r^4-kr^3-jr^2-ir+8.
	\]

	Let $U'$ be the set of integers constructed as above.
	Also, note that for any integer $a'\in U'$ we have $0<a'<b$.
	These integers were constructed so that the following statement holds.
	\begin{observation}[\cite{DBLP:journals/ipl/Woeginger97}]
	\label{obs:intcor}
		A set of four integers from $U'$ add up to $b$ if and only if they correspond to some elements $x_i\in X,y_j\in Y,z_k\in Z$ and tuple $t_{(i,j,k)}\in T$ where $t_{(i,j,k)}=(x_i,y_j,z_k)$.
	\end{observation}
%	\begin{proof}
%		(If) Suppose $x_i\in X,y_j\in Y,z_k\in Z$ and $t_{(i,j,k)}\in T$ where $t_{(i,j,k)}=(x_i,y_j,z_k)$ then it is easy to verify that indeed 
%		\[
%			t'_{(i,j,k)} + x'_i + y'_j + z'_k = b
%		\]
%		(Only if) Conversely, suppose that four integers $a'_1,a'_2,a'_3,a'_4$ sum to $b$.
%		Considering the equation modulo $r$ and using the fact $1+2+4+8$ is the only possible way of obtaining 15 as a sum of four elements (possibly with repetition) from the set $\{1,2,4,8\}$ and therefore we conclude the integers must correspond one element each from $X,Y,Z,T$.
%		This means we can write $a'_1,a'_2,a'_3,a'_4$ as $x'_i,y'_j,z'_k,t'_{(i_0,j_0,k_0)}$.
%		Now, considering the equation modulo $r^2,r^3,r^4$ gives us $i=i_0$,$j=j_0$, and $k=k_0$.
%	\end{proof}
	To obtain a Vector Bin Packing instance for each integer $a'\in U'$ construct a vector
	\[
		\mathbf a = \left(\frac{1}{5}+\frac{a'}{5b},\frac{3}{10}-\frac{a'}{5b}\right).
	\]
	We also construct additional $|T|+3q-4\beta(I_M)$ dummy vectors
	\[
		\mathbf d = \left(\frac{3}{5},\frac{3}{5}\right),
	\]
	where $\beta(\cdot)$ is a function from instances of 3-Dimensional Matching to positive integers.
	We now note a few properties of the vectors.
	\begin{observation}[\cite{DBLP:journals/ipl/Woeginger97, DBLP:journals/mor/BansalCKS06}]
	\label{obs:binsize}
		A bin can contain at most $4$ vectors.
		If a bin contains a dummy vector it can contain at most one more vector.
		Furthermore, a set of two vectors fit in a bin if and only if at least one of them is non-dummy.
	\end{observation}
%	\begin{proof}
%	The first part follows from the fact that the first component of any vector is strictly greater than $\frac{1}{5}$.
%	The second part of the claim follows from the fact any vector in the instance has first component greater than $\frac{1}{5}$ and the dummy vector has first component equal to $\frac{3}{5}$.
%		For the third part, observe if there are two dummy vectors then both the components would add up to $2\cdot \frac{3}{5}>1$.
%		However, if one of them is a non-dummy vector, then notice that both components of a non-dummy vector are less than $\frac{2}{5}$, and both components of any vector are less than $\frac{3}{5}$. Hence, both components of the sum are less than 1.
%	\end{proof}
	\begin{observation}[\cite{DBLP:journals/ipl/Woeginger97}]
	\label{obs:vectorcor}
		A set of four vectors fits in a bin if and only if it corresponds to a tuple.
	\end{observation}
%	\begin{proof}
%		(If) For a tuple $t_{(i,j,k)}=(x_i,y_j,z_k)$, we have $t_{(i,j,k)}'+x'_i+y'_j+z'_k=b$ by \Cref{lem:intcor}.
%		So, we have
%		\[
%			\mathbf t_{(i,j,k)} + \mathbf x_i + \mathbf y_j + \mathbf z_k 
%			= \left(\frac{4}{5}+\frac{t'_{(i,j,k)}+x'_i+y'_j+z'_k}{5b}, \frac{6}{5}-\frac{t'_{(i,j,k)}+x'_i+y'_j+z'_k}{5b}\right)
%			= (1,1).
%		\]
%		(Only if) Suppose there are four vectors $\mathbf a_1,\mathbf a_2,\mathbf a_3,\mathbf a_4$ which fit in a bin.
%		By \Cref{lem:binsize} all the vectors are non-dummy vectors. Hence, each vector can be written as:
%		\[
%			\mathbf a_i = \left(\frac{1}{5}+\frac{a'_i}{5b},\frac{3}{10}-\frac{a'_i}{5b}\right).
%		\]
%		As they fit in a bin we get the following two conditions,
%		\begin{align*}
%			\frac{4}{5}+\frac{\sum\limits_{i=1}^4a'_i}{5b}\leq 1,\\
%			\frac{6}{5}-\frac{\sum\limits_{i=1}^4a'_i}{5b}\leq 1,
%		\end{align*}
%		which simplify to $\sum\limits_{i=1}^4a'_i\leq b$ and $\sum\limits_{i=1}^4a'_i\geq b$.
%		Combining the inequalities we get $\sum\limits_{i=1}^4a'_i=b$.
%		Therefore, by \Cref{lem:intcor} the vectors correspond to a tuple.
%	\end{proof}
	Now we show that the above construction is a gap reduction from MAX-3-DM to 2-Dimensional Vector Bin Packing (cf. Theorem 2.1 from \cite{DBLP:journals/mor/BansalCKS06}).
	\begin{lemma}
	\label{lem:main}
	If a MAX-3-DM instance $I_M$ has a solution with $\beta(I_M)$ tuples then the constructed Vector Bin Packing instance has a solution with $|T|+3q-3\beta(I_M)$ bins.
	Otherwise, if all the solutions of the MAX-3-DM instance have at most $\alpha(I_M)$ tuples then the constructed instance needs at least $|T|+3q-\frac{\alpha(I_M)}{3}-\frac{8\beta(I_M)}{3}$ bins where $\alpha(\cdot)$ is any function from instances of MAX-3-DM to positive integers.
	\end{lemma}
	\begin{proof}
		First, we show that if a MAX-3-DM instance has a matching consisting of $\beta(I_M)$ tuples, then the Vector Bin Packing instance has a solution of $|T|+3q-3\beta(I_M)$ bins.
		Using \Cref{obs:vectorcor}, the $4\beta(I_M)$ vectors corresponding to the $\beta(I_M)$ tuples and their elements can be packed into $\beta(I_M)$ bins.
		Each of the remaining $|T|+3q-4\beta(I_M)$ non-dummy vectors can be packed along with a dummy vector into $|T|+3q-4\beta(I_M)$ bins by \Cref{obs:binsize}.
	
		Now, suppose that for a given instance all the solutions have at most $\alpha(I_M)$ tuples.
		Let $n_g$ be the number of bins with 4 vectors, $n_d$ be the number of bins with dummy vectors, and $n_r$ be the rest of the bins.
		Since any solution to the Vector Bin Packing instance must pack all the non-dummy vectors we have 
		\begin{enumerate}[(a)]
			\item any bin containing four vectors consists of only non-dummy vectors by \Cref{obs:vectorcor};
			\item any bin containing a dummy vector contains at most one non-dummy vector, by \Cref{obs:binsize}; 
			\item any other bin can contain at most 3 vectors by \Cref{obs:binsize}.
		\end{enumerate}
		Therefore, we have
		\begin{align*}
			4n_g+3n_r+n_d&\geq 3q+|T|.
		\end{align*}
		Now, by \Cref{obs:binsize} we have $n_d=|T|+3q-4\beta(I_M)$.
		Hence, the above inequality simplifies to
		\begin{align*}
			4n_g+3n_r&\geq 4\beta(I_M)\\
			\Rightarrow n_g+n_r&\geq \frac{4}{3}\beta(I_M)-\frac{n_g}{3}\\
			\Rightarrow n_g+n_r+n_d&\geq |T|+3q-\frac{n_g}{3}-\frac{8}{3}\beta(I_M)
		\end{align*}
		where the last inequality follows from $n_d = |T|+3q-4\beta(I_M)$.
	
		Since there are at most $\alpha(I_M)$ tuples in the MAX-3-DM instance, by \Cref{obs:vectorcor} we have $n_g\leq \alpha(I_M)$.
		Therefore, the number of bins needed is at least $|T|+3q-\frac{\alpha(I_M)}{3}-\frac{8\beta(I_M)}{3}$.
	\end{proof}
	The following inapproximability for Vector Bin Packing directly follows from \Cref{lem:main}.
	\begin{theorem}
		\label{thm:vbp}
		There is no APTAS for the $d$-Dimensional Vector Bin Packing problem with $d\geq 2$ unless P=NP.
		Furthermore, for the $d=2$ case there is no algorithm with asymptotic approximation ratio better than $\frac{600}{599}$.
	\end{theorem}
	\begin{proof}
		Suppose that there is an algorithm with approximation ratio $1+\frac{\beta_0-\alpha_0}{15-9\beta_0}$.
		Then we can distinguish between MAX-3-DM-E2 instances (i) having a solution of $\lceil \beta_0 q \rceil$ tuples and (ii) having no solutions with more than $\lfloor \alpha_0 q \rfloor$ tuples using \Cref{lem:main} with $\alpha(I_M)=\lfloor \alpha_0 q \rfloor$ and $\beta(I_M)=\lceil \beta_0 q \rceil$, hence solving GAP($\alpha_0, \beta_0)$-3-DM-E2.
		By \Cref{thm:matching}, GAP($\alpha_0, \beta_0$)-3-DM-E2 is NP-hard, where $\beta_0=0.979338843$, and $\alpha_0=0.9690082645$.
		Hence, we obtain the bound of $1+\frac{\beta_0-\alpha_0}{15-9\beta_0}$. 
		Simple calculations will show this is at least $1+\frac{1}{599}$.
	\end{proof}

	\section{The Original Proof}
	\label{sec:original}
	The original proof uses essentially the same reduction as ours, i.e., there we had $r=32q$, $b=r^4+15$ and then for each $x_i\in X,y_i\in Y,z_i\in Z$ we had
	\begin{align*}
		x'_i=ir+1,\\
		y'_i=ir^2+2,\\
		z'_i=ir^3+4,
	\end{align*}
	and for $t_l\in T$ was $t'_l$ defined by
	\[
		t'_l=r^4-kr^3-jr^2-ir+8.
	\]
	And finally, to obtain a Vector Bin Packing instance for each integer $a'$ constructed above construct the vector
	\[
		\mathbf a = \left(\frac{1}{5}+\frac{a'}{5b},\frac{3}{10}-\frac{a'}{5b}\right).
	\]
	The above set of vectors forms a 2-Dimensional Vector Bin Packing instance $\mathbf U$.
	A noticeable difference from our reduction being the absence of dummy vectors.
	Also note that $r=32q$ and tuples are denoted by $t_l$.
	In \cite{DBLP:journals/ipl/Woeginger97}, Woeginger claimed that
	\begin{claim*}[Observation 4 in \cite{DBLP:journals/ipl/Woeginger97}]
	Any set of 3 vectors in $\mathbf U$ can be packed in a bin. No set of 5 vectors in $\mathbf U$ can be packed into a bin.
	\end{claim*}
	We show that this claim does not hold in general. In particular, all sets of 3 vectors can not be packed into a bin.

	Consider the tuple vectors for the tuples $t_1={(x_1,y_1,z_1)}$, $t_2={(x_2,y_1,z_1)}$, and $t_3={(x_3,y_1,z_1)}$.
	According to the claim, the vectors $\mathbf t_1,\mathbf t_2, \mathbf t_3$ corresponding to the above tuples can be packed in a bin.
	Suppose $\mathbf t_1,\mathbf t_2,\mathbf t_3$ can indeed be packed in a bin.
	This implies that the first components of the vectors do not exceed 1, i.e.,
	\[
	\frac{3}{5}+\frac{t'_1+t'_2+t'_3}{5b}\leq 1,
	\]
	which simplifies to
	\[
	t'_1+t'_2+t'_3\leq 2b.
	\]
	Finally, using 
	\begin{align*}
		t'_1=r^4-r^3-r^2-r+8,\\
		t'_2=r^4-r^3-r^2-2r+8,\\
		t'_3=r^4-r^3-r^2-3r+8,
	\end{align*}
	and
	\[
		b=r^4+15,
	\]
	along with further simplification we get
	\[
		r^4\leq 3r^3+3r^2+6r+6.
	\]
	But this inequality does not even hold for $r\geq 32$ whereas 32 is the smallest value for $r=32q$.
	Thus, the claim is incorrect.

	\section{Vector Bin Packing with skewed items}

	In this section, we adapt the reduction presented in~\Cref{sec:revised} to show that any algorithm for $\delta$-Skewed $d$-Dimensional Vector Bin Packing cannot have an approximation ratio better than $1+\epsilon$ if $\delta>20\sqrt \epsilon$ for small values of $\epsilon$.

	Again, we start by defining a few integers based on the given MAX-3-DM instance $I_M$.
	Let $m=\lceil \frac{2}{\delta} \rceil - 1$, for some %\footnote{the choice $\delta\in (0,\frac{1}{3}]$ lets us avoid degenerate case with $m\leq 4$}
	$\delta \in (0,\frac{2}{5})$.
	Choose $n>m2^m$ and set $r=nq$ and $b=r^m+2^{m+1}-1$.
	Define integers $x'_i,y'_i,z'_i$ corresponding to ${x_i\in X},{y_i\in Y}, {z_i\in Z}$ to be
	\begin{align*}
		x'_i=ir+1,\\
		y'_i=ir^2+2,\\
		z'_i=ir^3+4,
	\end{align*}
	and for each $t_{(i,j,k)}=(x_i,y_j,z_k)\in T$ define $t'_{(i,j,k)}$ as
	\[
		t'_{(i,j,k)}=r^m-\sum_{l=4}^{m-1}r^l-kr^3-jr^2-ir+2^m.
	\]
	Finally, we add additional $|T|$ integers for each $l\in \{4,\dots,m-1\}$,
	\[
		c'_l=r^l+2^l.
	\]

	Let $U'$ be the set of integers constructed as above.
	As before, for any integer $a'\in U'$ we have $0<a'<b$ and the following statement holds.
	\begin{observation}
	\label{obs:skewintcor}
		A subset of integers $S\subseteq U'$ with $|S|=m$ adds up to $b$ if and only if there are $x'_i,y'_j,z'_k,t'_{(i,j,k)}\in S$ corresponding to some elements $x_i\in X,y_j\in Y,z_k\in Z$ and tuple $t_{(i,j,k)}\in T$ where $t_{(i,j,k)}=(x_i,y_j,z_k)$ and $c'_l\in S$ for each $l\in \{4,\dots,m-1\}$.
	\end{observation}
	To obtain a Vector Bin Packing instance for each integer $a'\in U'$, construct the vector
	\[
		\mathbf a = \left(\frac{1}{m+1}+\frac{a'}{(m+1)b},\frac{m+2}{m(m+1)}-\frac{a'}{(m+1)b}\right).
	\]
	We also construct additional $(m-3)|T|+3q-m\beta(I_M)$ dummy vectors 
	\[
		\mathbf d = \left(\frac{m-1}{m+1},0\right),
	\]
	where $\beta(\cdot)$ is again a function from instances of MAX-3-DM to positive integers which will be fixed later.
	Notice that each of these vectors has a dimension whose size is less than $\frac{2}{m+1}\leq \delta$.
	Again, we note a few properties of the vectors.
	\begin{observation}
	\label{obs:skewbinsize}
		A bin can contain at most $m$ vectors.
		If a bin contains a dummy vector it can contain at most one more vector.
		Furthermore, a set of two vectors fit in a bin if and only if at least one of them is non-dummy.
	\end{observation}
	\begin{observation}
	\label{obs:skewvectorcor}
		A set $S$ of $m$ vectors fits in a bin if and only if there are $\mathbf x_i,\mathbf y_j,\mathbf z_k,\mathbf t_{(i,j,k)}\in S$ corresponding to some elements $x_i\in X,y_j\in Y,z_k\in Z$ and tuple $t_{(i,j,k)}\in T$ where $t_{(i,j,k)}=(x_i,y_j,z_k)$ and $\mathbf c_l\in S$ for each $l\in \{4,\dots,m-1\}$.
	\end{observation}
	Now we show that the above construction is also a gap reduction from MAX-3-DM to 2-Dimensional Vector Bin Packing.
	\begin{lemma}
	\label{lem:skewmain}
		If a MAX-3-DM instance $I_M$ has a solution with $\beta(I_M)$ tuples then the constructed Vector Bin Packing instance has a solution with $(m-3)|T|+3q-(m-1)\beta(I_M)$ bins.
		Otherwise, if all the solutions of the MAX-3-DM instance have at most $\alpha(I_M)$ tuples then the constructed instance needs at least $(m-3)|T|+3q-\frac{\alpha(I_M)}{m-1}-\frac{m(m-2)\beta(I_M)}{m-1}$ bins where $\alpha(\cdot)$ is a function from MAX-3-DM instances to positive integers.
	\end{lemma}
	\begin{proof}
		First, we show that if a MAX-3-DM instance has a matching consisting of $\beta(I_M)$ tuples, then the Vector Bin Packing instance has a solution of $(m-3)|T|+3q-(m-1)\beta(I_M)$ bins.
		%As in \Cref{lem:main}, in this case we can pack $m\beta(I_M)$ vectors in $\beta(I_M)$ bins using \Cref{obs:skewvectorcor} while the rest of the vectors can be packed in $(m-3)|T|+3q-m\beta(I_M)$ bins by \Cref{obs:skewbinsize}.
		Using \Cref{obs:skewvectorcor}, the $m\beta(I_M)$ vectors corresponding to the $\beta(I_M)$ tuples and their elements and one vector $\mathbf c_l$ for each $l\in \{4,\dots,m-1\}$ can be packed into $\beta(I_M)$ bins.
		As in \Cref{lem:main}, by \Cref{obs:skewbinsize} we can pack the remaining vectors in $(m-3)|T|+3q-m\beta(I_M)$ bins.
		%Each of the $(m-3)|T|+3q-m\beta(I_M)$ remaining non-dummy vectors can be packed along with a dummy vector into $(m-3)|T|+3q-m\beta(I_M)$ bins by \Cref{obs:skewbinsize}.
	
		Now, suppose that for a given instance, all the solutions have at most $\alpha(I_M)$ tuples.
		Let $n_g$ be the number of bins with $m$ vectors, $n_d$ be the number of bins with dummy vectors, and $n_r$ be the rest of the bins.
		Now, since any solution to the bin packing instance must cover all the non-dummy vectors we have 
		\begin{enumerate}[(a)]
			\item any bin containing $m$ vectors consists of only non-dummy vectors by \Cref{obs:skewvectorcor};
			\item any bin containing a dummy vector contains at most one non-dummy vector, by \Cref{obs:skewbinsize};
			\item any other bin can contain at most $m-1$ vectors by \Cref{obs:skewbinsize}.
		\end{enumerate}
		Therefore, we have
		\begin{align*}
			mn_g+(m-1)n_r+n_d&\geq 3q+(m-3)|T|.
		\end{align*}
		Again, as in \Cref{lem:main}, we can simplify the above inequality using the facts: $n_d=(m-3)|T|+3q-m\beta(I_M)$ (by \Cref{obs:skewbinsize}); $n_g \leq \alpha(I_M)$ (by \Cref{obs:skewvectorcor}).
		Hence, we can conclude that the number of bins needed is at least
		% Hence, the above inequality simplifies to
	%	\begin{align*}
			%mn_g+(m-1)n_r\geq& m\beta(I_M)\\
			%\Rightarrow n_g+n_r\geq& \frac{m}{m-1}\beta(I_M)-\frac{n_g}{m-1}\\
	%		n_g+n_r+n_d\geq (m-3)|T|&+3q-\frac{n_g}{m-1}\\&-\frac{m(m-2)}{m-1}\beta(I_M).
	%	\end{align*}
	%	Since there are at most $\alpha(I_M)$ tuples in the MAX-3-DM instance, by \Cref{obs:skewvectorcor} we have $n_g\leq \alpha(I_M)$.
	%	Therefore, the number of bins needed is at least
		\begin{align*}
			&&(m-3)|T|+3q-\frac{\alpha(I_M)}{m-1}-\frac{m(m-2)\beta(I_M)}{m-1}.
		\end{align*}
		\end{proof}
		%The following inapproximability for Vector Bin Packing directly follows from \Cref{lem:skewmain}.
		\begin{theorem}
			For any $\epsilon\in (0,\frac{1}{2500})$ there is no $1+\epsilon$-approximation algorithm for the $\delta$-Skewed $d$-Dimensional Vector Bin Packing problem with $d\geq 2, \delta> 20\sqrt \epsilon$ unless P=NP.
		\end{theorem}
		\begin{proof}
			%Suppose that there is an algorithm with approximation ratio $1+\frac{\beta_0-\alpha_0}{m(2m-3)-(m-1)^2\beta_0}$.
			%Then we can distinguish between MAX-3-DM-E2 instances (i) having a solution of $\lceil \beta_0 q \rceil$ tuples and (ii) having no solutions with more than $\lfloor \alpha_0 q \rfloor$ tuples using \Cref{lem:skewmain} with $\alpha(I_M)=\lfloor \alpha_0 q \rfloor$ and $\beta(I_M)=\lceil \beta_0 q \rceil$, hence solving GAP($\alpha_0, \beta_0)$)-3-DM-E2.
			%By \Cref{thm:matching}, GAP($\alpha_0, \beta_0$)-3-DM-E2 is NP-hard, where $\beta_0=0.979338843$, and $\alpha_0=0.9690082645$.
			%Hence, we obtain the bound of $1+\frac{\beta_0-\alpha_0}{m(2m-3)-(m-1)^2\beta_0}$. 
			Using \Cref{thm:matching} and \Cref{lem:skewmain} along with arguments used in the proof of \Cref{thm:vbp}, we obtain the bound of $1+\frac{\beta_0-\alpha_0}{m(2m-3)-(m-1)^2\beta_0}$, where $\alpha_0,\beta_0$ are the parameters from \Cref{thm:matching}.
			Simple calculations will show that this is strictly greater than $1+\frac{\delta^2}{400}$.
			Using $\delta=20\sqrt{\epsilon}$ we get the desired result.
		\end{proof}

	\section{Vector Bin Covering has no APTAS}
\label{sec:vbc}
In this section, we prove that Vector Bin Covering has no APTAS unless P=NP by adapting the proof presented in \Cref{sec:revised}.
The analysis is slightly more complicated and bears some resemblance to the analysis of the reduction to the Geometric Bin Covering problem presented in~\cite{DBLP:journals/jda/ChlebikC09}.
Again, we obtain a gap preserving reduction from MAX-3-DM to 2-Dimensional Vector Bin Covering.
We start with the same set of integers $U'$ we had in \Cref{sec:revised}.
To obtain a Vector Bin Covering instance for each integer $a'$ in $U'$, construct the vector
\[
	\mathbf a = \left(\frac{1}{5}+\frac{a'}{5b},\frac{3}{10}-\frac{a'}{5b}\right).
\]
We also construct additional $|T|+3q-4\beta(I_M)$ dummy vectors 
\[
	\mathbf d = \left(\frac{9}{10},\frac{9}{10}\right),
\]
where $\beta(\cdot)$ is a function from instances of 3-Dimensional Matching to positive integers (note that the size of dummy vectors is different from \Cref{sec:revised}).
If a bin has at least one dummy vector then we call it a D-bin.
Otherwise, if a bin has no dummy vectors the we call it a non-D-bin.
Again, we note a few properties of the constructed vectors.
\begin{observation}
	\label{obs:coversize}
	Any set of 5 vectors can cover a bin.
	Any vector along with a dummy vector can cover a bin.
	At least 2 vectors are needed to form a bin.
\end{observation}
\begin{observation}
	\label{obs:vectorcor2}
	A set of four vectors covers a non-D-bin if and only if it corresponds to a tuple.
\end{observation}
Now we are ready to prove our main lemma showing our reduction is indeed a gap preserving reduction.
\begin{lemma}
	\label{lem:gapvbc}
	If a MAX-3-DM instance $I_M$ has a solution with $\beta(I_M)$ tuples then there is a solution to the Vector Bin Covering instance with $|T|+3q-3\beta(I_M)$ tuples.
	Otherwise, if all the solutions of $I_M$ have at most $\alpha(I_M)$ tuples then the constructed instance can cover at most $|T|+3q-\frac{16}{5}\beta(I_M)+\frac{\alpha(I_M)}{5}$ bins, where $\alpha(\cdot)$ is a function from MAX-3-DM instances to positive integers.
\end{lemma}
\begin{proof}
	Suppose that $\text{OPT}(I_M)\geq \beta(I_M)$.
	As in \Cref{lem:main}, we use the optimal solution and \Cref{obs:vectorcor2} to cover $\beta(I_M)$ bins using $4\beta(I_M)$ non-dummy vectors
	while $|T|+3q-4\beta(I_M)$ bins are covered with the remaining vectors using \Cref{obs:coversize}.
	%Then we can cover $\beta(I_M)$ bins using $4\beta(I_M)$ vectors corresponding to the $\beta(I_M)$ tuples from the solution of $I_M$ using \Cref{obs:vectorcor2}.
	%Now, we have $|T|+3q-4\beta(I_M)$ non-dummy vectors left along with exactly $|T|+3q-4\beta(I_M)$ dummy vectors.
	%These vectors, by \Cref{obs:coversize}, can cover $|T|+3q-4\beta(I_M)$ bins.

	Now, suppose that every solution of the MAX-3-DM instance has value at most $\alpha(I_M)$.
	Consider an optimal solution to the constructed Vector Bin Covering instance.
	We can normalize an optimal solution without any loss in the number of bins covered as follows.
	\begin{enumerate}[(a)]
		\item
			\emph{Number of dummy vectors equals the number of D-bins.}
			To that end, observe that there are $|T|+3q-4\beta(I_M)$ dummy vectors and $|T|+3q$ non-dummy vectors.
			Therefore, by \Cref{obs:coversize} there is a solution with $|T|+3q-4\beta(I_M)$ bins.
			Hence, an optimal solution must have at least $|T|+3q-4\beta(I_M)$ bins.
			%Clearly, more than $|T|+3q-4\beta(I_M)$ bins are covered in an optimal solution.
			%Note that the required condition can be rephrased as:
			%\[\text{number of dummy vectors} = \text{number of D-bins}.\]
			Suppose there is a bin (say $B_1$) with at least two dummy vectors, i.e., number of dummy vectors $>$ number of D-bins.
			Then we can show there is another optimal solution with a larger number of D-bins.
			As there are at least $|T|+3q-4\beta(I_m)-1$ bins which still need to be covered and $|T|+3q-4\beta(I_m)-2$ dummy vectors remaining, there must be at least one non-D-bin (say $B_2$).
			%If there is a bin with at least two dummy vectors then there must be a non-D-bin as at least 2 vectors are needed to cover a bin (\Cref{obs:coversize}).
			Now, note that by \Cref{obs:coversize} $B_2$ must contain at least 2 vectors.
			%Therefore, we can pick one non-D-bin for each bin with 2 dummy vectors, which must contain 2 non-dummy vectors by \Cref{obs:coversize}.
			Again by \Cref{obs:coversize}, we can now exchange one vector from $B_2$ with a dummy vector in $B_1$ to obtain a another solution with same number of bins while increasing the number of D-bins.
			%Similar arguments can be used for bins having $k$ dummy vectors, i.e., there are at least $(k-1)n$ non-D-bins and then similar rearangements can be done to obtain $kn$ bins with one dummy vector.
		\item 
			\emph{No subset of a bin can cover a bin.}
			To that end, some non-dummy vectors can be left out, i.e., they may be designated as not covering any bin.
			For non-D-bin we just choose any minimal subset of the vectors that can cover the bin.
			Now, if the solution satisfies condition (a) then by \Cref{obs:coversize} each D-bin contains one dummy vector and at least one non-dummy vectors; hence, we can keep the dummy vector and one non-dummy vector.
			%Also, note that the number of such vectors can be at most four as five vectors always form a bin.
	\end{enumerate}
	Let $n_d$ be the number of D-bins, $n_g$ be the number of non-D-bins covered by 4 vectors and $n_r$ be the number of non-D-bin covered by 5 vectors.
	Again, as in \Cref{lem:main}, we use the facts: $n_d=|T|+3q-4\beta(I_M)$ (due to our normalization); $3q+|T|$ non-dummy vectors; $n_g\leq \alpha(I_M)$ (by \Cref{obs:vectorcor2}) to get
	%\begin{align*}
%		n_d+4n_g+5n_r&\leq 3q+|T|\\
%		\Rightarrow 4n_g+5n_r&\leq 4\beta(I_M)\\%&&[\text{Since, } n_d=|T|+3q-4\beta(I_M)]\\
%		\Rightarrow n_g+n_r&\leq \frac{4}{5}\beta(I_M)+\frac{n_g}{5}.
%	\end{align*}
%	By \Cref{obs:vectorcor2}, we have $n_g\leq \alpha(I_M)$. 
%	Therefore,
	\begin{align*}
		%n_g+n_r&\leq \frac{4}{5}\beta(I_M)+\frac{\alpha(I_M)}{5}\\
		n_d+n_g+n_r&\leq |T|+ 3q-\frac{16}{5}\beta(I_M)+\frac{\alpha(I_M)}{5}.
	\end{align*}
	In other words, the number of bins covered is at most
	\begin{align*}
		&&|T|+3q-\frac{16}{5}\beta(I_M)+\frac{\alpha(I_M)}{5}.
	\end{align*}
\end{proof}
\begin{theorem}
	There is no APTAS for $d$-Dimensional Vector Bin Covering with $d\geq 2$ unless P=NP.
	Furthermore, for the 2-Dimensional Vector Bin Covering there is no algorithm with asymptotic approximation ratio better than $\frac{998}{997}$.
\end{theorem}
\begin{proof}
	Using \Cref{thm:matching} and \Cref{lem:gapvbc} along with arguments used in the proof of \Cref{thm:vbp}, we obtain the bound of $1+\frac{\beta_0-\alpha_0}{25-16\beta_0+\alpha_0}$, where $\alpha_0,\beta_0$ are the parameters from \Cref{thm:matching}.
	Simple calculations will show this is at least $1+\frac{1}{997}$.
\end{proof}

	\section*{Concluding Remarks}
\begin{enumerate}
	\item The definition of the class APX-hard is quite technical (see~\cite{DBLP:conf/coco/Crescenzi97}) and our reduction does not show Vector Bin Packing is in APX-hard (despite ruling out an asymptotic PTAS).
		The same is also true for the Vector Bin Covering reduction.
	\item There is still a considerable gap between the best-known approximation ratio of $\frac{4}{3}+\varepsilon$ for the 2-Dimensional Vector Bin Packing problem and our lower bound for it.
	Similarly, there is a large gap in case of the 2-Dimensional Vector Bin Covering problem.
\end{enumerate}

	\section*{Acknowledgements}
	I would like to thank Arindam Khan for introducing me to multidimensional bin packing problems and all the discussions we had on them.
I would also like to thank KVN Sreenivas, Rameesh Paul, Eklavya Sharma, and anonymous referees for their valuable comments.
	\bibliography{references}
	\appendix
\end{document}